\RequirePackage{fix-cm}
\documentclass[smallextended]{svjour3}       
\smartqed  
\usepackage{url}
\usepackage{float}
\usepackage{graphicx}
\usepackage{longtable}
\usepackage{subfigure}
\usepackage{soul,color}

\usepackage{algorithm}
\usepackage{algpseudocode}
\usepackage{amsfonts}
\usepackage{amsmath}
\newtheorem{obs}[theorem]{Observation}

\newtheorem{lem}[theorem]{Lemma}
\newtheorem{thm}[theorem]{Theorem}

\def\be{\begin{equation}}
\def\ee{\end{equation}}

\def\bea{\begin{eqnarray}}
\def\eea{\end{eqnarray}}
\def\nn{\nonumber}

\def\ba{\begin{array}}
\def\ea{\end{array}}
\def\reff#1{(\ref{#1})}

\begin{document}

\title{Approximate the individually fair $k$-center with outliers}


\titlerunning{Approximate the individually fair $k$-center with outliers}        

\author{Lu Han \and  Dachuan Xu \and \\ Yicheng Xu\footnote{Corresponding author. Email:{yc.xu@siat.ac.cn}} \and  Ping Yang
}

\authorrunning{L. Han, D. Xu, Y. Xu, and P. Yang} 

\institute{L. Han \at School of Science, Beijing University of Posts and Telecommunications, Beijing 100876, P.R. China.
\and D. Xu \at Beijing Institute for Scientific and Engineering Computing, Beijing University of Technology, Beijing 100124, P.R. China.
\and Y. Xu  \at Shenzhen Institute of Advanced Technology, Chinese Academy of Sciences, Shenzhen 518055, P.R. China.
\and P. Yang \at College of Statistics and Data Science, Beijing University of
Technology, Beijing 100124, P.R. China.}
\date{Received: date / Accepted: date}
\maketitle

\begin{abstract}
In this paper, we propose and investigate the individually fair $k$-center with outliers (IF$k$CO). In the IF$k$CO, we are given an $n$-sized vertex set in a metric space,
as well as integers $k$ and $q$. At most $k$ vertices can be selected as the centers and at most $q$ vertices can be selected as the outliers. The centers are selected to serve all the not-an-outlier (i.e., served) vertices. The so-called individual fairness constraint restricts that every served vertex must have a selected center not too far way. More precisely, it is supposed that there exists at least one center among its  $\lceil (n-q) / k \rceil$ closest neighbors for every served vertex. Because every center serves $(n -q) / k$ vertices on the average.
The objective is to select centers and outliers, assign every served vertex to some center, so as to minimize the maximum fairness ratio over all served vertices, where the fairness ratio of a vertex is defined as the ratio between its distance with the assigned center and
its distance with a $\lceil (n - q )/k \rceil_{\rm th}$ closest neighbor. As our main contribution, a 4-approximation algorithm is presented, based on which we develop an improved algorithm
from a practical perspective.
Extensive experiment results on both synthetic datasets and real-world datasets are presented to illustrate the effectiveness of the proposed algorithms.

\keywords{$k$-center \and Individual fairness \and Outliers \and Approximation algorithm }
\end{abstract}

\section{Introduction}
Clustering problems are studied due to their widespread
applications in operations research and machine learning areas \cite{bprst,cgts,g,hs85,hs86,l,sta}. As a consequence, some natural and significant variants also attract lots of research interests \cite{c,kps,ks,kls,xxdw,xxzz}.

The concept of fairness is introduced into clustering problems very recently.  Chierichetti et al. \cite{cklv} first studies the fairness in the sense that each cluster is required to have approximately equal proportion of representations. Many more explanation of fairness in the clustering problems are proposed since then.
They vary from each other in considering fairness in different objects. Some consider the balance between clusters \cite{aekm,biosvw,bcfn,bgkkrss,hjv,sss}, some consider the balance within selected centers \cite{jnn,kam} and others consider the balance of cost functions \cite{abv,gsv,mv}. All these fairness can be viewed as the so-called group fairness, and very limited work concentrates on the  individual fairness.

The individual fairness is proposed by Jung et al. \cite{jkl} in the sense of population density. They study the individually fair $k$-center (IF$k$C) where an $n$-sized vertex set in a metric space and an integer $k$ are given. At most $k$ vertices can be selected as the centers to serve all the given vertices. A vertex would expect that there exists a center among its $\lceil n / k \rceil$  closest neighbors, since  each open center serves $n/k$ vertices on the average. Jung et al. \cite{jkl} show that sometimes it is impossible to find the suitable centers which satisfy the expectation of each vertex. So the IF$k$C focuses on optimizing how far from the ideal expectation. Specifically, the objective of the IF$k$C is to select at most $k$ vertices as centers, and assign each vertex to some center, so as to minimize the maximum fairness ratio over all vertices, where the fairness ratio of a vertex is defined as the ratio between its distance with the assigned center and its distance with a $\lceil n /k \rceil_{\rm th}$ closest neighbor.
Jung et al. \cite{jkl} give a $4$-approximation algorithm for the IF$k$C. Soon afterwards, under the notion of individual fairness in \cite{jkl}, Mahabadi and Vakilian \cite{mvi} and Vakilian and  Yal{\c{c}}{\i}ner \cite{vy} study the $k$-clustering with $l_p$-norm cost function.

However, an isolated vertex may cause huge loss of the overall clustering quality in the IF$k$C. It is of significance to overcome this shortcoming of the problem. Towards this end, we introduce the individually fair $k$-center with outliers (IF$k$CO) in this paper, which allows some vertices, called outliers, to be discarded when clustering. Thus, an additional integer $q$ is given.
At most $q$ vertices can be selected as the outliers which could not be served. If a vertex is selected as an outlier, it does not care the distance between a center and itself. If a vertex is not an outlier, it would expect that there exists a center among its $\lceil (n-q) / k \rceil$ closest neighbors, since ideally we wish that each center serves $(n -q) / k$ vertices. The goal is to select at most $k$ centers and at most $q$ outliers, assign each not-an-outlier vertex to some center, so as to minimize the maximum outlier-related fairness ratio over all not-an-outlier vertices, where the outlier-related fairness ratio of a vertex is defined as the ratio between its distance with the assigned center and its distance with a $\lceil (n-q) /k \rceil_{\rm th}$ closest neighbor.
Our contributions are fourfold.
\begin{itemize}
\item Contribution 1: We first present a naive but natural algorithm for the IF$k$CO and prove that the algorithm may return a solution far from being optimal.

\item Contribution 2: After finding out the naive algorithm's principle of selecting centers is lack of rationality, we then design a basic $4$-approximation algorithm for the IF$k$CO, which successfully avoids the shortcoming of the naive algorithm.

\item Contribution 3: Unfortunately, the basic algorithm has its own limitation that it may select very few vertices as outliers. We further propose a refined $4$-approximation algorithm to deal with the limitation.

\item Contribution 4: We apply the refined algorithm to several instances and show that the refined algorithm is well-behaved.

\end{itemize}

The remainder of this paper is structured as follows. In section 2, the mathematical description of the IF$k$CO is given, followed by a naive algorithm. In section 3, the main part of this paper, we present two algorithms for the IF$k$CO, a basic one and a refined one. In section 4, we test the refined algorithm on a large scale of synthetic and real-world instances. In section 5, we discuss the practical aspect of the proposed algorithms as well as some interesting directions.

\section{Preliminaries}
We start with the mathematical  descriptions of the IF$k$CO and IF$k$C. A naive attempt show that the algorithm for the IF$k$C can easily obtain a feasible solution for the IF$k$CO instance. However, it can be arbitrarily bad.

\subsection{Problem descriptions}
In any instance for the IF$k$CO, denoted by $\mathcal{I}_{{\rm IF}k{\rm CO}}$, we are given a vertex set $V$ with size $n$. Let $d_{ij}$ be the distance between a pair of vertices $(i, j)$ with $i, j \in V$. It is assumed that the distances are metric, i.e., obey the following assumptions.
\begin{itemize}
\item  They are \emph{non-negative}, i.e., $d_{ij} \geq 0$ for any $i, j \in V$;
\item They are \emph{symmetric}, i.e., $d_{ii}=0$ and $d_{ij} = d_{ji}$ for any $i, j \in V$;
\item  They satisfy the \emph{triangle inequality}, i.e., $d_{hi} + d_{ij} \geq d_{hj}$ for any $h, i, j \in V$.
\end{itemize}
Also, we are given the integers $k$  and $q$, the maximum number of vertices that can be selected as the centers and that of the outliers. For each $i \in V$,
let $NR_q(i)$ be the distance between $i$ and its $\lceil (n - q )/k \rceil_{\rm th}$ nearest neighbor. Note that any vertex itself is its nearest neighbor.
We call $NR_q(i)$ the outlier-related neighborhood radius of $i$.
The aim is to select  vertices $S \subseteq V$ as centers and $O \subseteq V$ as outliers, assign each vertex $i \in V \setminus O$ to some center $\sigma(i) \in S$, such that $|S| \leq k$, $|O|\leq q$, and the maximum ratio of $d_{\sigma(i)i} / NR_q(i)$ of a vertex in $V \setminus O$ is minimized.

We use $(S, O, \sigma)$ to denote a solution for the IF$k$CO instance $\mathcal{I}_{{\rm IF}k{\rm CO}}$, in which $S \subseteq V$ is the set of selected centers, $O \subseteq V$ is the set of selected outliers and $\sigma: V \setminus O \rightarrow S $ is an assignment mapping each vertex in $V \setminus O$ to some center in $S$. A solution $(S, O, \sigma)$ is feasible if $|S| \leq k$ and $|O| \leq q$. For each vertex  $i \in V \setminus O$, we call $d_{\sigma(i)i} / NR_q(i)$ its outlier-related fairness ratio. For the solution $(S, O, \sigma)$, we call $\alpha(S, O, \sigma)$ its outlier-related fairness ratio, which is
the maximum outlier-related fairness ratio of a vetex in $V \setminus O$, i.e.,
$$\alpha(S, O, \sigma) = \max \limits_{i \in V \setminus O} \frac{d_{\sigma(i)i}}{NR_q(i)}.$$
Denote by $(S^*, O^*, \sigma^*)$ the optimal solution for  $\mathcal{I}_{{\rm IF}k{\rm CO}}$, and  $OPT_{\mathcal{I}_{{\rm IF}k{\rm CO}}}$ the outlier-related fairness ratio of $(S^*, O^*, \sigma^*)$, i.e.,
$$(S^*, O^*, \sigma^*) = \arg\min \limits_{(S, O, \sigma) : |S| \leq k, |O| \leq q} \alpha (S, O, \sigma),$$
$${\rm and} ~ OPT_{\mathcal{I}_{{\rm IF}k{\rm CO}}} = \alpha(S^*, O^*, \sigma^*) = \max \limits_{i \in V \setminus O^*} \frac{d_{\sigma^*(i)i}}{NR_q(i)}.$$

By setting $q=0$, the  $\mathcal{I}_{{\rm IF}k{\rm CO}}$ reduces to an IF$k$C instance. More specifically, in an IF$k$C instance $\mathcal{I}_{{\rm IF}k{\rm C}}$,
we are given a vertex set $V$. Each pair of vertices $(i, j)$, where $i, j \in V$, has a distance $d_{ij}$. We assume that the distances are non-negative, symmetric, and satisfy the triangle inequality. Also, we are given an integer $k$, the maximum number of vertices that can be selected as the centers.
For each vertex $i \in V$,
let $NR(i)$ be the distance between $i$ and its $\lceil {n }/{k} \rceil_{\rm th}$ nearest neighbor. We call $NR(i)$ the neighborhood radius of $i$. The goal is to select  vertices $S \subseteq V$ as centers, assign each vertex $i \in V$ to some center $\sigma(i) \in S$, such that $|S| \leq k$, and the maximum ratio of $d_{\sigma(i)i}/NR(i)$ of a vertex in $V$ is minimized.

We use $(S, \sigma)$ to denote a solution for the IF$k$C instance $\mathcal{I}_{{\rm IF}k{\rm C}}$, in which $S \subseteq V$ is the set of selected centers and $\sigma: V  \rightarrow S $ is an assignment mapping each vertex in $V$ to some center in $S$. A solution $(S, \sigma)$ is feasible if $|S| \leq k$. For each vertex  $i \in V$, we call $d_{\sigma(i)i}/NR(i)$ its fairness ratio. For the solution $(S, \sigma)$, we call $\alpha(S, O)$ its fairness ratio, which is the maximum fairness ratio of a vetex in $V$, i.e.,
$$\alpha(S, \sigma) = \max \limits_{i \in V} \frac{d_{\sigma(i)i}}{NR(i)}.$$

\subsection{An attempt}
Herein, we present a naive but quite natural algorithm that is able to give a feasible solution for $\mathcal{I}_{{\rm IF}k{\rm CO}}$. For the instance, first remove its input of $q$ in order to obtain an IF$k$C instance $\mathcal{I}_{{\rm IF}k{\rm C}}$. Then, use the algorithm for the IF$k$C to solve $\mathcal{I}_{{\rm IF}k{\rm C}}$ and obtain a feasible solution for $\mathcal{I}_{{\rm IF}k{\rm CO}}$. The naive algorithm is shown as Algorithm 1. It is worth mentioning that Step 2 of Algorithm 1 is a slightly modified version of the 2FAIRKCENTER algorithm for the IF$k$C appeared in \cite{jkl}.

\begin{algorithm} \label{alg1}
\caption{: A Naive Algorithm for the IF$k$CO.}
{ {\bf Input:} An  IF$k$CO instance $\mathcal{I}_{{\rm IF}k{\rm CO}}= (V,  \{d_{ij}\}_{i, j \in V}, k, q)$.}\\
{ {\bf Output:} A feasible solution $(S, O, \sigma)$ for the instance  $\mathcal{I}_{{\rm IF}k{\rm CO}}$.}
\begin{description}
\item[Step 1]  For $\mathcal{I}_{{\rm IF}k{\rm CO}}$, get rid of $q$ to yield a IF$k$C instance $\mathcal{I}_{{\rm IF}k{\rm C}}= (V,  \{d_{ij}\}_{i, j \in V}, k)$.

\item[Step 2]
{Initially, set $P:= V$, $S: = \emptyset$}.
\begin{description}
\item ~~~{\bf While} $P \not= \emptyset$ {\bf do}
\begin{description}
\item Find a vertex $s \in P$ such that
\begin{description}
\item $s := \arg \min \limits_{i \in P} NR(i).$
\end{description}
\item Update $S:=  S \cup \{s\}$, $P:= \{ i \in P: d_{is} > 2 \cdot NR(i)\}$.
\end{description}
\end{description}

\item[Step 3]
Set $O := \emptyset$, $\sigma(i):= \arg \min_{h \in S} d_{ih}$ for each $i \in V$.
\item[Step 4] Output $(S, O, \sigma)$ as the solution for the instance  $\mathcal{I}_{{\rm IF}k{\rm CO}}$.
\end{description}
\end{algorithm}

For any selected center $s \in S$, denote by $V(s, NR(s))$ the set of vertices within the distance of $NR(s)$ from $s$. We call $V(s, NR(s))$ the neighboring vertex set of $s$.
Here are some observations about Algorithm 1.
\begin{obs} \label{o1}
For any selected center $s \in S$, there are at least $n/k$ vertices in its neighboring vertex set.
\end{obs}
This observation can be seen from the definition of $NR(s)$.
\begin{obs} \label{o2}
If a vertex $s$ is selected as a center, any other vertex in its neighboring vertex set cannot be a center.
\end{obs}
\begin{proof}
 When a vertex $s$ is selected as a center, each vertex  $i \in V(s, NR(s))$ either already be removed from the current $P$ or it satisfies $d_{is} \leq NR(s) \leq 2NR(s) \leq 2 NR(i)$. The last inequality follows by the principle of selecting centers in Step 2.
For the second case, the vertex $i$ will be removed from $P$ and cannot be a center anymore, because of the principle of updating $P$ in Step 2.
\qed
\end{proof}
\begin{obs}  \label{o3}
For any two selected centers  $s, s^\prime \in S$, their neighboring vertex set are disjoint.
\end{obs}
\begin{proof}
Assume that center $s$ is selected before $s^\prime$. If there exists a vertex $i \in V(s,NR(s)) \cap V(s^\prime,NR(s^\prime))$, we have that $d_{ss^\prime} \leq d_{si} + d_{is^\prime} \leq NR(s) + NR(s^\prime) \leq 2 NR(s^\prime)$. The last inequality follows by the principle of selecting centers in Step 2. In this case, because of the principle of updating $P$ in Step 2, once $s$ is selected, the vertex $s^\prime$ will be remove from $P$ and cannot be a center anymore, which is a contradiction.
\qed
\end{proof}

The following lemma gives the feasibility of the solution $(S, O, \sigma)$ returned from Algorithm 1.


\begin{lem} \label{lem1}
Algorithm 1 outputs a feasible solution $(S, O, \sigma)$ for any {\rm IF}$k${\rm CO} instance $\mathcal{I}_{{\rm IF}k{\rm CO}}$.
\end{lem}

\begin{proof}
Recall that a solution $(S, O, \sigma)$ is feasible only if $|S| \leq k$ and $|O| \leq q$.
Since $O = \emptyset$, the cardinality bound of $O$ obviously holds. We only need to prove $|S| \leq k$.

From Observations \ref{o1}-\ref{o3}, we can see that each iteration of Step 2 of Algorithm 1 ensures that at least $n/k$ disjoint vertices are removed from the current $P$. Once $P=\emptyset$, we end the iterations. Therefore, there are at most $k$ iterations in Step 2, which implies $|S| \leq k$.
This completes the proof of the lemma.
\qed
 \end{proof}

Note that Algorithm 1 may return a feasible solution far from the optimal one for some IF$k$CO instances, as shown by Example 1. \\
{\bf Example 1.} Consider the IF$k$CO instance $(V,  \{d_{ij}\}_{i, j \in V}, k, q)$ where $V=\{h, i, j\}$, $d_{hi}=M$, $d_{hj}=M$, $d_{ij}=1$, $k=1$ and $q=1$. Suppose that $M >1$.

Recall that for any $v \in V$, its neighborhood radius $NR(v)$ is the distance between $v$ and its $\lceil n/k \rceil$ nearest neighbor. Since $n=|\{h, i, j\}|=3$ and $\lceil n/k \rceil = \lceil 3/1 \rceil = 3$, the neighborhood radiuses used in Algorithm 1 are $NR(h)=NR(i)=NR(j)=M$. If we use Algorithm 1 to solve the instance, Algorithm 1 will arbitrarily select a vertex in $V$ as the center. It is possible that Algorithm 1 selects $h$ as the center and outputs $(\{h\}, \emptyset, \sigma_1)$ as the solution, where $\sigma_1(h)=\sigma_1(i)=\sigma_1(j)=h$. Recall that for any $v \in V$, its outlier-related neighborhood radius $NR_q(v)$ is the distance between $v$ and its $\lceil (n -q)/{k} \rceil$ nearest neighbor. Since $\lceil (n-q)/k \rceil = \lceil (3-1)/ 1 \rceil = 2$, the outlier-related neighborhood radiuses are $NR_q(h)=M$ and $NR_q(i)=NR_q(j)=1$.
Therefore, the outlier-ralated fairness ratio of the solution $(\{h\}, \emptyset, \sigma_1)$ is $M$, i.e.,
\bea
\alpha(\{h\}, \emptyset, \sigma_1) &=& \max \limits_{v \in \{h, i, j\} \setminus \emptyset} \frac{d_{\sigma_1(v)v}}{NR_q(v)} \nn \\
&=&\max\{\frac{d_{\sigma_1(h)h}}{NR_q(h)},\frac{d_{\sigma_1(i)i}}{NR_q(i)},
\frac{d_{\sigma_1(j)j}}{NR_q(j)}\} \nn\\
&=&\max\{\frac{d_{hh}}{NR_q(h)},\frac{d_{hi}}{NR_q(i)},
\frac{d_{hj}}{NR_q(j)}\} \nn\\
&=&\max\{\frac{0}{M},\frac{M}{1},\frac{M}{1}\} \nn \\
&=&M. \nn
\eea
The optimal solution is to select either $i$ or $j$ as the center and $h$ as the outlier. Assume that the selected center is $i$. Therefore, the optimal solution is $(\{i\}, \{h\}, \sigma^*)$, where $\sigma^*(i)=\sigma^*(j)=i$.
The outlier-related fairness ratio of the solution $(\{i\}, \{h\}, \sigma^*)$ is $1$, i.e.,
\bea
\alpha(\{i\}, \{h\}, \sigma^*) &=& \max \limits_{v \in \{h, i, j\} \setminus \{h\}} \frac{d_{\sigma^*(v)v}}{NR_q(v)} \nn \\
&=&\max\{\frac{d_{\sigma^*(i)i}}{NR_q(i)},\frac{d_{\sigma^*(j)j}}{NR_q(j)}\} \nn\\
&=&\max\{\frac{d_{ii}}{NR_q(i)},
\frac{d_{ij}}{NR_q(j)}\} \nn\\
&=&\max\{\frac{0}{1},\frac{1}{1}\} \nn \\
&=&1. \nn
\eea
Therefore, we have  $$\frac{\alpha(\{h\}, \emptyset, \sigma_1)}{\alpha(\{i\}, \{h\}, \sigma^*)} =\frac{M}{1},$$ which implies that the Algorithm 1 may return a solution far from being optimal. An illustration of Example 1 is given in Fig. \ref{fig1}.

\begin{figure}[h]
\centering
  {
    \label{fig1a} 
    \includegraphics[height=4.5cm]{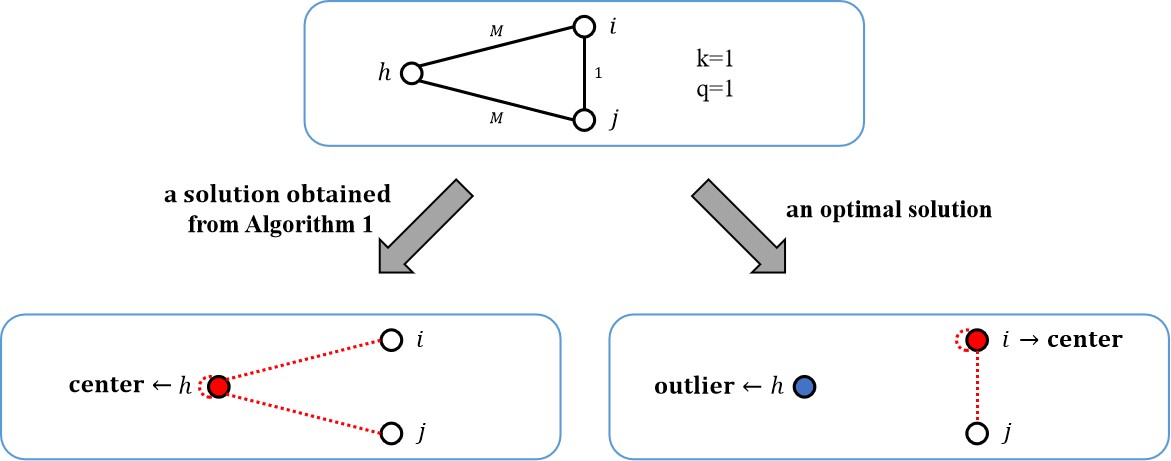}}
\caption{An illustration of Example 1. In the top graph, the circles are the vertices and the distances are given alongside the vertex pairs. The red and blue circles represent the selected centers and outlier, respectively. The dotted lines represent the assignments of the vertices.   }
\label{fig1}
\end{figure}

\section{Advisable algorithms for the IF$k$CO}
In this section, we first propose a basic $4$-approximation algorithm for the IF$k$CO. Then, we give a refined algorithm, which overcomes the limitation of the basic one.
\subsection{A basic algorithm}

The main adversary for Algorithm 1 is that to select vertex $i$ with a minimum neighborhood radius $NR(i)$ may lead to very bad outcome,  and no vertices are selected as outliers in the obtained solution. Therefore, we specifically design a basic algorithm for the IF$k$CO. Our basic algorithm keeps finding a selectable vertex $i$ with the minimum outlier-related neighborhood radius $NR_q(i)$ as a center while the set of selectable vertices is not empty and the number of currently chosen centers is less than $k$.
The basic algorithm is formally presented as Algorithm 2.
\begin{algorithm} \label{alg1}
\caption{: A Basic Algorithm for the IF$k$CO.}

{ {\bf Input:} An  IF$k$CO instance $\mathcal{I}_{{\rm IF}k{\rm CO}}= (V,  \{d_{ij}\}_{i, j \in V}, k, q)$.}\\
{ {\bf Output:} A feasible solution $(S, O, \sigma)$ for the instance  $\mathcal{I}_{{\rm IF}k{\rm CO}}$.}

\begin{description}

\item[Step 1]
{Initially, set $P:= V$, $S: = \emptyset$}.
\begin{description}
\item ~~~{\bf While} $P \not= \emptyset$, $|S| < k $ {\bf do}
\begin{description}
\item Find a vertex $s \in P$ such that
\begin{description}
\item  $s := \arg \min \limits_{i \in P} NR_q(i).$
\end{description}
\item Update $S:=  S \cup \{s\}$, $P:= \{ i \in P: d_{is} > 2 \cdot NR_q(i)\}$.
\end{description}
\end{description}

\item[Step 2]
Set $O := P$, $\sigma(i):= \arg \min_{h \in S} d_{ih}$ for each $i \in V \setminus O$.
\item[Step 3] Output $(S, O, \sigma)$ as the solution for the instance  $\mathcal{I}_{{\rm IF}k{\rm CO}}$.
\end{description}
\end{algorithm}

For any selected center $s \in S$, denote by $V(s, NR_q(s))$ the set of vertices within the distance of $NR_q(s)$ from $s$. We call $V(s, NR_q(s))$ the outlier-related neighboring vertex set of $s$.
Here are some observations about Algorithm 2.
\begin{obs} \label{o4}
For any selected center $s \in S$, there are at least $(n - q )/k$ vertices in its  outlier-related neighboring vertex set.
\end{obs}
This observation can be seen from the definition of $NR_q(s)$.
\begin{obs} \label{o5}
If a vertex $s$ is selected as a center, any other vertex in its outlier-related neighboring vertex set cannot be a center.
\end{obs}
\begin{proof}
When a vertex $s$ is selected as a center, each vertex  $i \in V(s, NR_q(s))$ either already be removed from the current $P$ or it satisfies $d_{is} \leq NR_q(s) \leq 2NR_q(s) \leq 2 NR(i)$. The last inequality follows by the principle of selecting centers in Step 1.
For the second case, the vertex $i$ will be removed from $P$ and cannot be a center anymore, because of the principle of updating $P$ in Step 1.
\qed
\end{proof}
\begin{obs}  \label{o6}
For any two selected centers  $s, s^\prime \in S$, their outlier-related neighboring vertex sets are disjoint.
\end{obs}
\begin{proof}
Assume that center $s$ is selected before $s^\prime$. If there exists a vertex $i \in V(s,NR_q(s)) \cap V(s^\prime,NR_q(s^\prime))$, we have that $d_{ss^\prime} \leq d_{si} + d_{is^\prime} \leq NR_q(s) + NR_q(s^\prime) \leq 2 NR_q(s^\prime)$. The last inequality follows by the principle of selecting centers in Step 1. In this case, because of the principle of updating $P$ in Step 1, once $s$ is selected, the vertex $s^\prime$ will be remove from $P$ and cannot be a center anymore, which is a contradiction.
\qed
\end{proof}

The following lemma gives the feasibility of the solution $(S, O, \sigma)$ obtained from Algorithm 2.

\begin{lem} \label{lem2}
Algorithm 2 outputs a feasible solution $(S, O, \sigma)$ for any {\rm IF}$k${\rm CO} instance $\mathcal{I}_{{\rm IF}k{\rm CO}}$.
\end{lem}
\begin{proof}
Recall that a solution $(S, O, \sigma)$ is feasible only if $|S| \leq k$ and $|O|\leq q.$
Note that Step 1 of Algorithm 2 guarantees that $|S| \leq k$. Thus we only need to prove  $|O| \leq q$.

We consider the two cases that may terminate Step 1. The simpler case is  $P = \emptyset$, and the other case is $|S|=k$. If $P = \emptyset$, the cardinality bound of $O$ obviously holds, since $|O|=|P|=0 \leq q$. If $|S|=k$, there are $k$ iterations.
We conclude  from  Observations \ref{o4}-\ref{o6} that each iteration of Step 1 in Algorithm 2 guarantees that at least $(n - q )/k$ disjoint vertices are removed from the current $P$. Therefore, the number of vertices removed from the initial $P$ is at least $n-q$. We have that $|O| = |P| \leq n - (n-q)= q$.
This completes the proof of the lemma.
\qed
\end{proof}


\begin{lem} \label{l2a}
The outlier-related fairness ratio of the solution $(S, O, \sigma)$ obtained from Algorithm 2 for any {\rm IF}$k${\rm CO} instance $\mathcal{I}_{{\rm IF}k{\rm CO}}$  is at most $2$, i.e.,
$$\alpha(S, O, \sigma) \leq 2.$$

\end{lem}
\begin{proof}
For the solution $(S, O, \sigma)$ obtained from Algorithm 2, from Step 1 of Algorithm 2, it can be seen that for each vertex $i \in V \setminus O$, there must exist a center $s \in S$ such that $d_{is} \leq 2 NR_q(i)$. Recall that
$\sigma(i):= \arg \min_{h \in S} d_{ih}$ for each $i \in V \setminus O$. Therefore, we have that
$$d_{\sigma(i)i} \leq d_{is} \leq 2 NR_q(i).$$
That means,
$$\frac{d_{\sigma(i)i}}{NR_q(i)} \leq 2 {\rm ~for~any}~ i \in V\setminus O.$$
Thus, we obtain that
\bea
&&\alpha(S, O, \sigma)= \max \limits_{i \in V \setminus O} \frac{d_{\sigma(i)i}}{NR_q(i)} \leq 2.\nn
\eea
This completes the proof of this lemma.
\qed
\end{proof}

\begin{lem} \label{l2b}
The outlier-related fairness ratio of the optimal solution $(S^*, O^*, \sigma^*)$ for any {\rm IF}$k${\rm CO} instance $\mathcal{I}_{{\rm IF}k{\rm CO}}$ is at least $1/2$, i.e.,
$$\alpha(S^*, O^*, \sigma^*)\geq \frac{1}{2}.$$
\end{lem}

\begin{proof}
For each center  $s^* \in S^*$ in the optimal solution $(S^*, O^*, \sigma^*)$, denote by $D(s^*)$ the set of vertices assigned to $s^*$ under the
assignment of $\sigma^*$, i.e.,
$$D(s^*) = \{i \in V \setminus O^*: \sigma^*(i) = s^*\}.$$
Since $(S^*, O^*, \sigma^*)$ is a feasible solution, we have that $|S^*| \leq k$, that $|O^*|\leq q$, and that $$|\bigcup_{s^* \in S^*}D(s^*)| = \sum \limits_{s^* \in S^*}|D(s^*)| =|V \setminus O^*|=|V|-|O^*| \geq n-q.$$
Therefore,
 $$\frac{|\sum \limits_{s^* \in S^*} D(s^*)|}{|S^*|} \geq \frac{n-q}{k}.$$
There must exist some center $s^\prime \in S^*$ satisfying  $|D(s^\prime)| \geq (n-q) / k.$ Let $s_f$ be the vertex farthest from $s^\prime$ in $D(s^\prime)$ under the assignment of $\sigma^*$.
Note that for any vertex $i \in D(s^\prime)$, we have that $$d_{is_f} \leq d_{is^\prime}+d_{s^\prime s_f} \leq 2 \cdot d_{s^\prime s_f}.$$
Since each vertex $i$ in $D(s^\prime)$ is within the distance of $2d_{s^\prime s_f}$ from $s_f$ and $|D(s^\prime)|\geq (n-q) / k$, combining with the definition of $NR_q(s_f)$, we obtain that $$NR_q(s_f) \leq 2 \cdot d_{s^\prime s_f}.$$
Thus, it  satisfies for $s_f \in V \setminus O^*$ that
$$\frac{d_{\sigma^*(s_f)s_f}}{NR_q(s_f)} = \frac{d_{s^\prime s_f}}{NR_q(s_f)} \geq \frac{1}{2}.$$
Therefore,
\bea
&&\alpha(S^*, O^*, \sigma^*)= \max \limits_{i \in V \setminus O^*} \frac{d_{\sigma^*(i)i}}{NR_q(i)} \geq \frac{d_{\sigma^*(s_f)s_f}}{NR_q(s_f)} \geq \frac{1}{2}.\nn
\eea
Complete the proof.
\qed
\end{proof}

From Lemmas \ref{l2a} and \ref{l2b}, for the solution $(S, O, \sigma)$ obtained from Algorithm 2 and the optimal solution $(S^*, O^*, \sigma^*)$, we have that $$\alpha(S, O, \sigma) \leq 4 \cdot \alpha(S^*, O^*, \sigma^*) = 4 \cdot OPT_{\mathcal{I}_{{\rm IF}k{\rm CO}}},$$ which implies the following result of Algorithm 2.
\begin{thm}\label{thm2}
Algorithm 2 is a $4$-approximation algorithm for the {\rm IF}$k${\rm CO}.
\end{thm}

Suppose that the Algorithm 2 is running on Example 1. It  will arbitrarily select $i$ or $j$ as the center and leave $h$ as an outlier, since Algorithm 2 keeps searching a selectable vertex $v$ with a minimum outlier-related neighborhood radius $NR_q(v)$ to select as a center and $NR_q(h)=M > 1=NR_q(i)=NR_q(j).$ Therefore, Algorithm 2 outputs an optimal solution for the instance.

\subsection{A refined algorithm}

A limitation of Algorithm 2 is that it may select very few vertices as outliers. To overcome this shortcoming, we present a refined algorithm which uses a binary search on a parameterized version of Algorithm 2. The refined algorithm is formally described in Algorithm 3. Compared with Algorithm 2, an additional parameter $l$ needs to be given as an input in Algorithm 3. The integer $l$ limits the number of iterations for searching a solution with a better outlier-related fairness ratio. The more steps of the iterations, it is more likely that we obtain a smaller ratio.


\begin{algorithm} \label{alg3}
\caption{: A Refined Algorithm for the IF$k$CO.}
{ {\bf Input:} An  IF$k$CO instance $\mathcal{I}_{{\rm IF}k{\rm CO}}= (V,  \{d_{ij}\}_{i, j \in V}, k, q)$, an integer $l \geq 0$.}\\
{ {\bf Output:} A feasible solution $(S, O, \sigma)$ for the instance  $\mathcal{I}_{{\rm IF}k{\rm CO}}$.}

\begin{description}

\item[Step 1]
Use Algorithm 2 to solve $\mathcal{I}_{{\rm IF}k{\rm CO}}$ and obtain a solution $(S_b, O_b, \sigma_b)$.
\item[Step 2]
Initially set $t:=0$, $\beta_1:=1$, $\beta_2 := 2$, $\beta:=\beta_1$ and $(S, O, \sigma):=(S_b, O_b, \sigma_b)$.
\item[Step 3]
{\bf While} $t < l$ {\bf do}

\begin{description}
\item \begin{description}
\item ~~~~Set $P_\beta:=V$, $S_\beta:=\emptyset$.
\item {\bf While} $P_\beta \not= \emptyset$, $|S_\beta| < k $ {\bf do}

\begin{description}
\item Find a vertex $s \in P_\beta$ such that

\begin{description}
\item  $s := \arg \min \limits_{i \in P_\beta} NR_q(i).$
\end{description}

\item Update $S_\beta:=  S_\beta \cup \{s\}$, $P_\beta:= \{ i \in P_\beta: d_{is} > \beta \cdot NR_q(i)\}$.

\end{description}

\item Set $O_\beta := P_\beta$, $\sigma_\beta(i):= \arg \min_{h \in S_\beta} d_{ih}$ for each $i \in V \setminus O_\beta$.
\item {\bf If} $|O_\beta| > q$ {\bf then}
\begin{description}
\item Update $\beta_1 := \beta$, $\beta := (\beta_1 + \beta_2) / 2$,   $t:=t+1$.
\end{description}
\item {\bf If} $|O_\beta| \leq q$ {\bf then}
\begin{description}
\item Update $(S, O, \sigma) :=(S_\beta, O_\beta, \sigma_\beta)$, $\beta_2 := \beta$, $\beta := (\beta_1 + \beta_2) / 2$,  $t:=t+1$.
\end{description}

\end{description}
\end{description}

\item[Step 4] Output $(S, O, \sigma)$ as the solution for the instance  $\mathcal{I}_{{\rm IF}k{\rm CO}}$.
\end{description}

\end{algorithm}

The following lemma gives the feasibility of the solution $(S, O, \sigma)$ obtained from Algorithm 3.
\begin{lem} \label{lem3}
Algorithm 3 outputs a feasible solution $(S, O, \sigma)$ for any {\rm IF}$k${\rm CO} instance $\mathcal{I}_{{\rm IF}k{\rm CO}}$.
\end{lem}
\begin{proof}
Recall that a solution $(S, O, \sigma)$ is feasible if $|S| \leq k$ and $|O| \leq q$. Initially, Algorithm 3 sets $(S, O, \sigma)$ as $(S_b, O_b, \sigma_b)$ which is a feasible solution obtained from Algorithm 2. Then, the principle of updating the solution $(S, O, \sigma)$ in Step 3 guarantees that the currently updated solution $(S_\beta, O_\beta, \sigma_\beta)$ must satisfy that $|S_\beta| \leq k$ and $|O_\beta| \leq q$. Therefore, the solution obtained from Algorithm 3 is a feasible solution.
\qed
\end{proof}


\begin{lem} \label{l3a}
The outlier-related fairness ratio of the solution $(S, O, \sigma)$ obtained from Algorithm 3 for any {\rm IF}$k${\rm CO} instance $\mathcal{I}_{{\rm IF}k{\rm CO}}$  is at most $2$, i.e.,
$$\alpha(S, O, \sigma) \leq 2.$$

\end{lem}
\begin{proof}
If Step 3 of Algorithm 3 does not update the solution $(S, O, \sigma)$, we output $(S_b, O_b, \sigma_b)$ as the final solution.
Therefore, from Lemma 3, we have that
\bea
\alpha(S, O, \sigma)&=&\alpha(S_b, O_b, \sigma_b) \leq 2.\label{th7a}
\eea
Now consider the case that Step 3 of Algorithm 3 updates the solution $(S, O, \sigma)$. Assume that the final updated solution is $(S_{\beta_f}, O_{\beta_f}, \sigma_{\beta_f})$. For any $i \in V \setminus O = V \setminus O_{\beta_f}$, from Step 3, we have that
$$d_{\sigma(i)i}= d_{\sigma_{\beta_f}(i)i} \leq \beta_f NR_q(i).$$
That means,
$$\frac{d_{\sigma(i)i}}{NR_q(i)} \leq \beta_f {\rm ~for~any}~ i \in V\setminus O.$$
Thus, we obtain that
\bea
&&\alpha(S, O, \sigma)= \max \limits_{i \in V \setminus O} \frac{d_{\sigma(i)i}}{NR_q(i)} \leq \beta_f \leq 2.\label{th7b}
\eea
The last inequality follows by the update principle of $\beta$ in Step 3, which guarantees that $\beta_f \in [1, 2]$.
Combining inequalities \reff{th7a} and \reff{th7b}, we complete the proof of this lemma.
\qed
\end{proof}

From Lemmas \ref{l2b} and \ref{l3a}, for the solution $(S, O, \sigma)$ obtained from Algorithm 3 and the optimal solution $(S^*, O^*, \sigma^*)$, we have that $$\alpha(S, O, \sigma) \leq 4 \cdot \alpha(S^*, O^*, \sigma^*) = 4 \cdot OPT_{\mathcal{I}_{{\rm IF}k{\rm CO}}},$$ which implies the following main result of Algorithm 3.
\begin{thm}\label{thm2}
Algorithm 3 is a $4$-approximation algorithm for the {\rm IF}$k${\rm CO}.
\end{thm}

Intuitively, it would be better to employ a parameterized algorithm for obtaining a solution with better outlier-related fairness ratio. We provided a parameterized version of Algorithm 1 in Algorithm 4. In the next section, we compare the performance of Algorithm 3 and 4 on a real-world dataset.

\begin{algorithm} \label{alg4}
\caption{: A Parameterized Version of Algorithm 1.}
{ {\bf Input:} An  IF$k$CO instance $\mathcal{I}_{{\rm IF}k{\rm CO}}= (V,  \{d_{ij}\}_{i, j \in V}, k, q)$, an integer $l \geq 0$.}\\
{ {\bf Output:} A feasible solution $(S, O, \sigma)$ for the instance  $\mathcal{I}_{{\rm IF}k{\rm CO}}$.}

\begin{description}

\item[Step 1]
Use Algorithm 1 to solve $\mathcal{I}_{{\rm IF}k{\rm CO}}$ and obtain a solution $(S_a, O_a, \sigma_a)$.
\item[Step 2]
Initially set $t:=0$, $\beta_1:=1$, $\beta_2 := 2$, $\beta:=\beta_1$ and $(S, O, \sigma):=(S_a, O_a, \sigma_a)$.
\item[Step 3]
{\bf While} $t < l$ {\bf do}

\begin{description}
\item \begin{description}
\item ~~~~Set $P_\beta:=V$, $S_\beta:=\emptyset$.
\item {\bf While} $P_\beta \not= \emptyset$ {\bf do}

\begin{description}
\item Find a vertex $s \in P_\beta$ such that

\begin{description}
\item  $s := \arg \min \limits_{i \in P_\beta} NR(i).$
\end{description}

\item Update $S_\beta:=  S_\beta \cup \{s\}$, $P_\beta:= \{ i \in P_\beta: d_{is} > \beta \cdot NR(i)\}$.

\end{description}

\item Set $O_\beta := \emptyset$, $\sigma_\beta(i):= \arg \min_{h \in S_\beta} d_{ih}$ for each $i \in V$.
\item {\bf If} $|S_\beta| > k$ {\bf then}
\begin{description}
\item Update $\beta_1 := \beta$, $\beta:=(\beta_1 + \beta_2) / 2$, $t:=t+1$.
\end{description}
\item {\bf If} $|S_\beta| \leq k$ {\bf then}
\begin{description}
\item Update $(S, O, \sigma) :=(S_\beta, O_\beta, \sigma_\beta)$, $\beta_2 := \beta$, $\beta:=(\beta_1+\beta_2) / 2$, $t:=t+1$.
\end{description}

\end{description}
\end{description}

\item[Step 4] Output $(S, O, \sigma)$ as the solution for the instance  $\mathcal{I}_{{\rm IF}k{\rm CO}}$.
\end{description}

\end{algorithm}

\section{Experiments}

In this section, we provide the experimental results of Algorithm 3 running on both synthetic datasets and real-world datasets to illustrate its effectiveness. The environment for the
experiments is Intel(R) Core(TM) CPU i7-6700 @3.40GHz with 8GB memory.

\subsection{Synthetic Datasets}
Theoretically, we prove that both Algorithm 2 and Algorithm 3 have a same approximation ratio of 4. However, Algorithm 3 has a much better performance in experiments.
In this subsection, we mainly test the proposed algorithms on synthetic datasets.

We randomly generate nine IF$k$CO instances with different settings of  $n$, $k$ and $q$. The nine instances are divided into three groups. Each group of the three aims to find out the effect of one parameter on the outlier-related fairness ratios of the solutions of Algorithm 3. The details of the settings are:
\begin{itemize}
\item  Group 1: Three randomly generated IF$k$CO instances, in which $k=20$, $q=50$ and $n=200, 1000, 5000$, respectively;
\item  Group 2: Three randomly generated IF$k$CO instances, in which $n=1000$, $q=50$ and $k=5, 20, 100$, respectively;
\item  Group 3: Three randomly generated IF$k$CO instances, in which $n=1000$, $k=20$, and  $q=20, 50, 100$, respectively.
\end{itemize}
For the instances in Group 1, 2 and 3, we show the continuous changing lines of the number of outliers selected by Algorithm 3 with respect to different values of $\beta$ in Fig. \ref{en}, \ref{ek} and \ref{eq}, respectively. The corresponding outlier-related fairness ratios of the solutions obtained from Algorithm 3 are also shown in Fig. 2. From our intuition, we belive that all the changing lines tend to decrease. It can be seen from Fig. 2 that in general they do, but some of the changing lines are not strictly decreasing.
Here are some specific observations for each group.
\begin{itemize}
\item  For Group 1: The positions of the three lines are consistent with our expectation. The larger $n$ is, the higher the position of the line, since for the same $\beta$ a larger $n$ would cause Algorithm 3 to output more outliers.
\item  For Group 2: The positions of the three lines are unconventional. We intuitively think that for the same $\beta$, a larger $k$ would cause Algorithm 3 to output a smaller number of outliers. But for the randomly generated instance where $k=5$, whatever $\beta$ is, the number of its outliers obtained from Algorithm 3 is always the smallest.
\item  For Group 3: The positions of the three lines are reasonable. The larger $q$ is, the smaller the corresponding outlier-related radiuses of all the vertices are. Smaller outlier-related radiuses would cause Algorithm 3 to output more outliers.
\end{itemize}
More remarkably, we find that Algorithm 3 is a very well-performed algorithm. For all the tested instances, the maximum outlier-related fairness ratio obtained from Algorithm 3 is only $1.31$, which is far more below the theoretical bound of 2.

\begin{figure}[h]
\centering
     \subfigure[The effect of $n$. ]
     {\label{en} 
    \includegraphics[height=4.2cm]{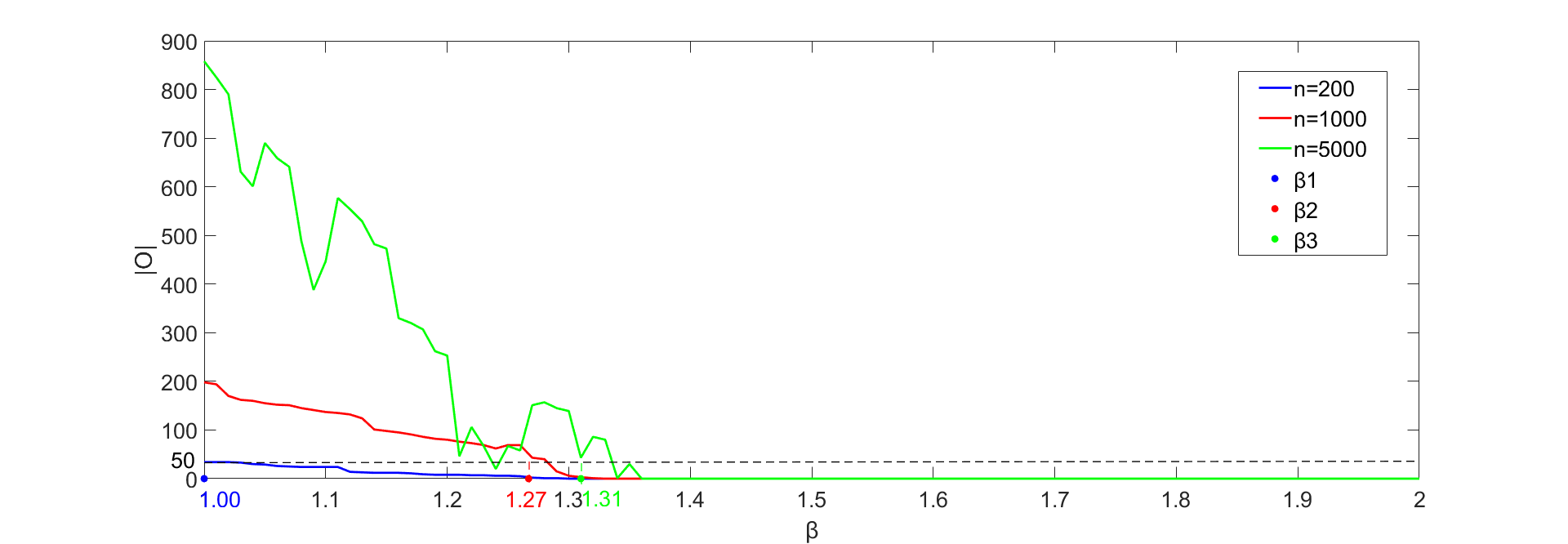}}

     \subfigure[The effect of $k$.]
     {\label{ek} 
    \includegraphics[height=4.2cm]{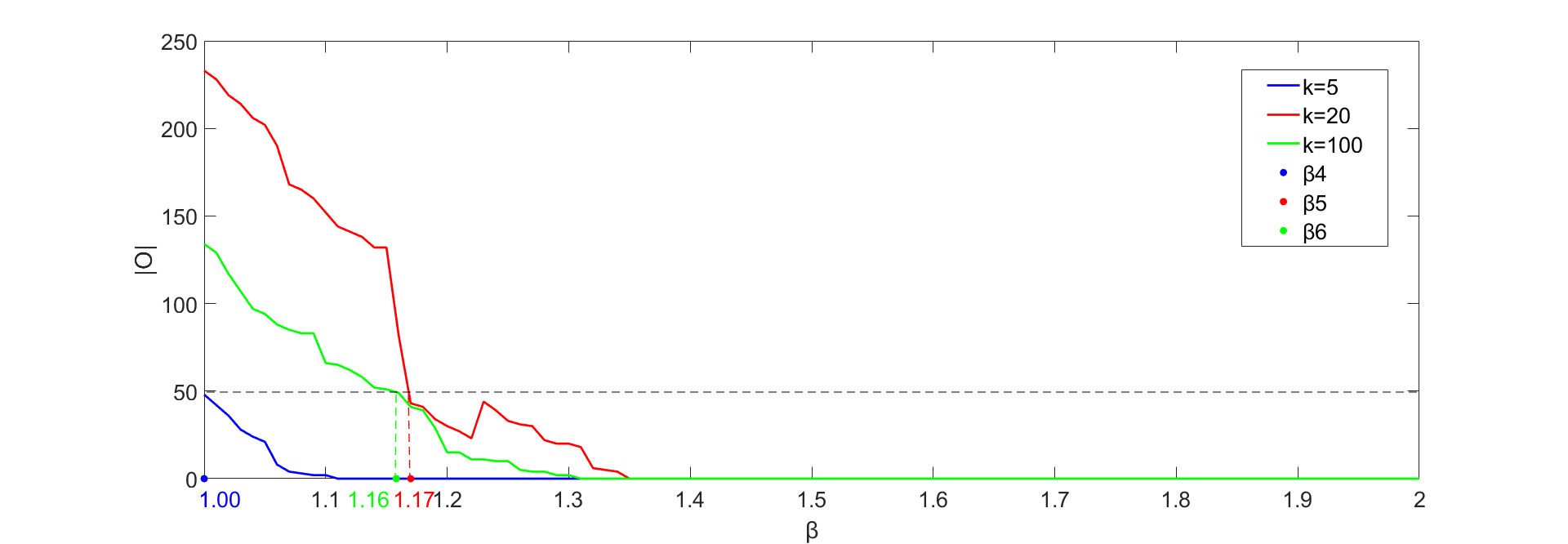}}

     \subfigure[The effect of $q$.]
  {\label{eq} 
    \includegraphics[height=4.2cm]{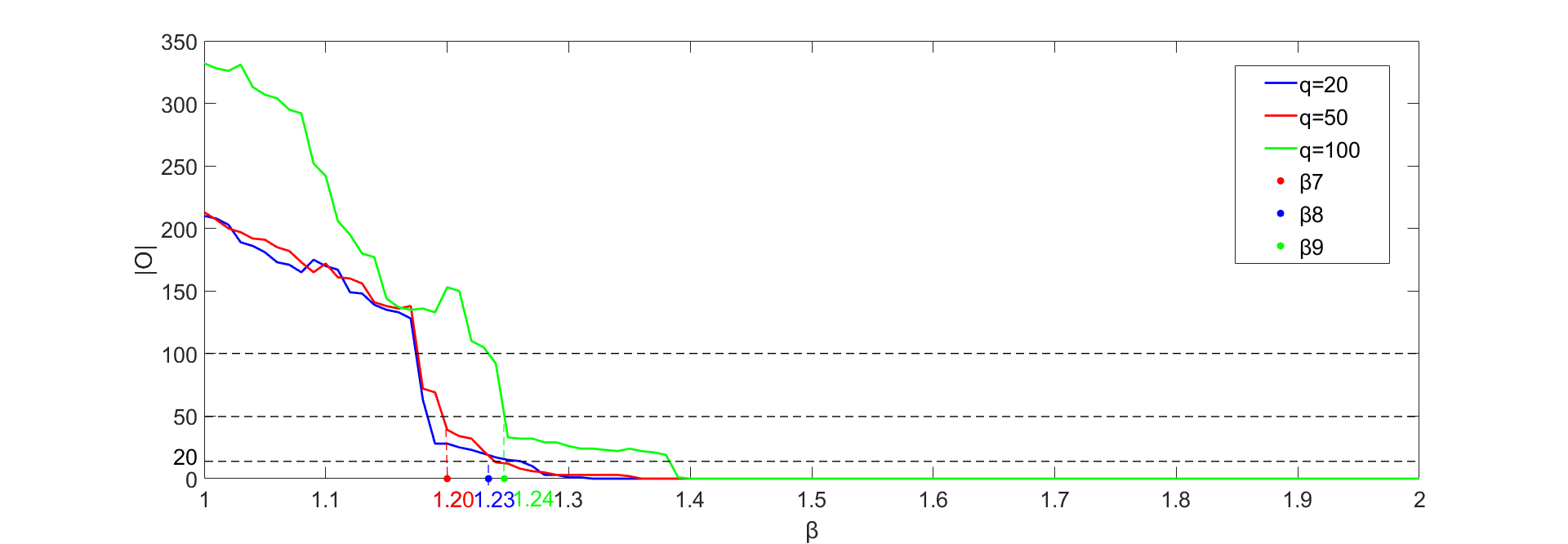}}
\caption{An illustration of the effect of the input of the IF$k$CO on the output of Algorithm 3. Each instance is represented by a colored line, which reflects the change of the number of outliers selected by Algorithm 3 with respect to different values of $\beta$. The colored dots represent the corresponding outlier-related fairness ratios of the solutions obtained from Algorithm 3.}
\end{figure}

\subsection{Real-world Datasets}
In this subsection, we test Algorithm 3 and 4 on the Shenzhen POI (Point of Interest) dataset collected from the open API of Gaode Maps \cite{link}. The target POI type contains $2936$ points, and since the POI type does not affect the results by any means, we hide it throughout this paper. The distances between any two points are measured in Euclidean distance after mapping the latitude and longitude of all the points onto a plane.

Recall that Algorithm 3 and 4 are the parameterized versions of Algorithm 2 and 1, respectively. It can be seen that Algorithm 2 performs much better than Algorithm 1 for Example 1. Thus intuitively, we would expect that the performance of Algorithm 3 is probably better than Algorithm 4 for the same instance.
 We show the centers selected by Algorithm 3 and 4 in Fig. \ref{alg1t3} and \ref{alg2t4}, respectively.
It turns out that Algorithm 3 is more likely to locate a center in dense areas compared with Algorithm 4. In other words, the points that are not likely to be outliers are of more importance in Algorithm 3 than in Algorithm 4. As a consequence, Algorithm 3 tends to obtain a smaller outlier-related fairness ratio than Algorithm 4.
This phenomenon makes sense because Algorithm 3 and Algorithm 4 keep searching the vertex $i$ with a minimum radius of $NR_q(i)$ and $NR(i)$ as a center, and a vertex $i$ in dense area is more likely to have the minimum $NR_q(i)$ than $NR(i)$.

\begin{figure}[h]
\centering
     \subfigure[The centers selected by Algorithm 3.]
     {\label{alg1t3} 
    \includegraphics[height=5.8cm]{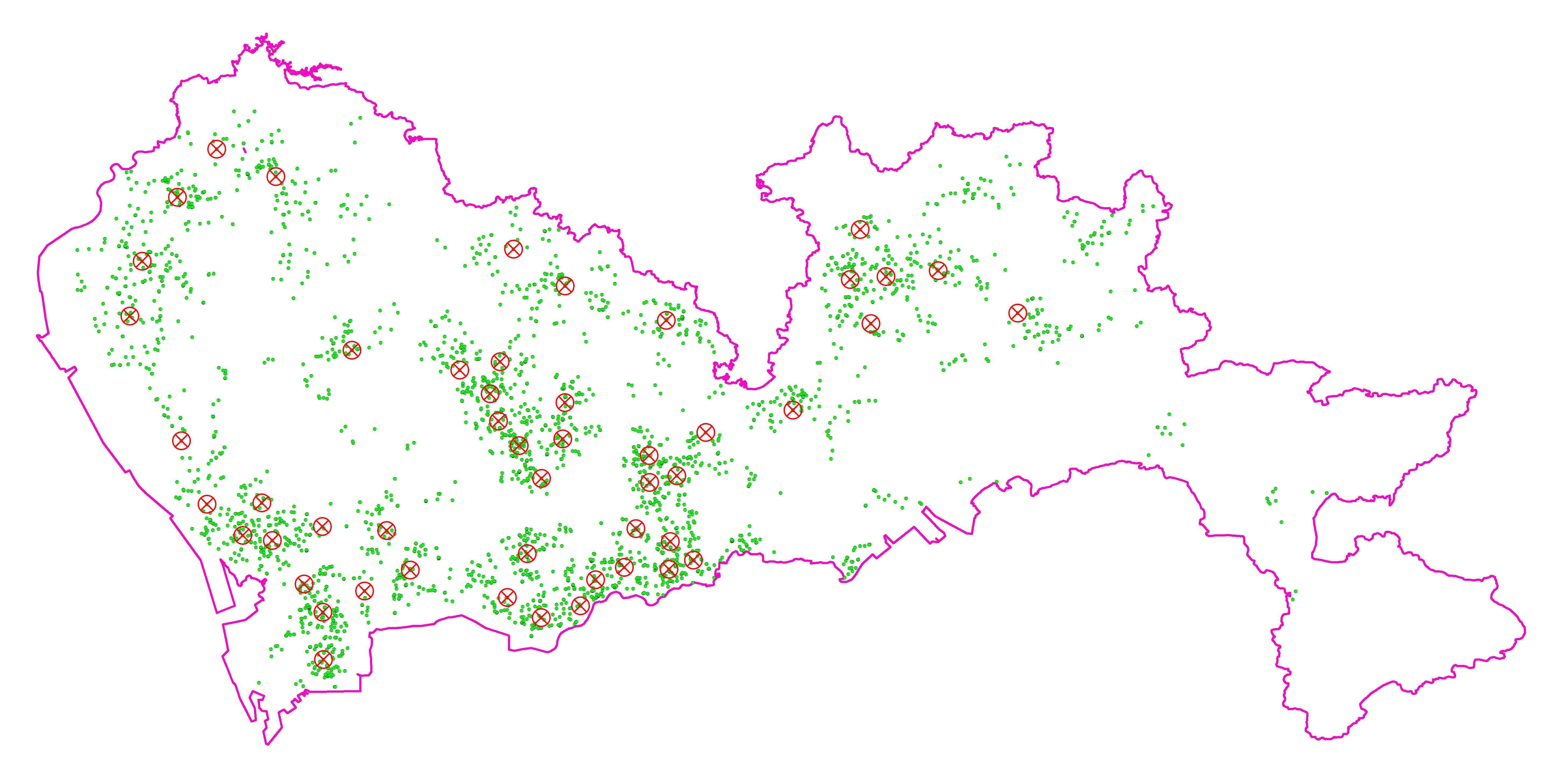}}
     \subfigure[The centers selected by Algorithm 4.]
  {\label{alg2t4} 
    \includegraphics[height=5.8cm]{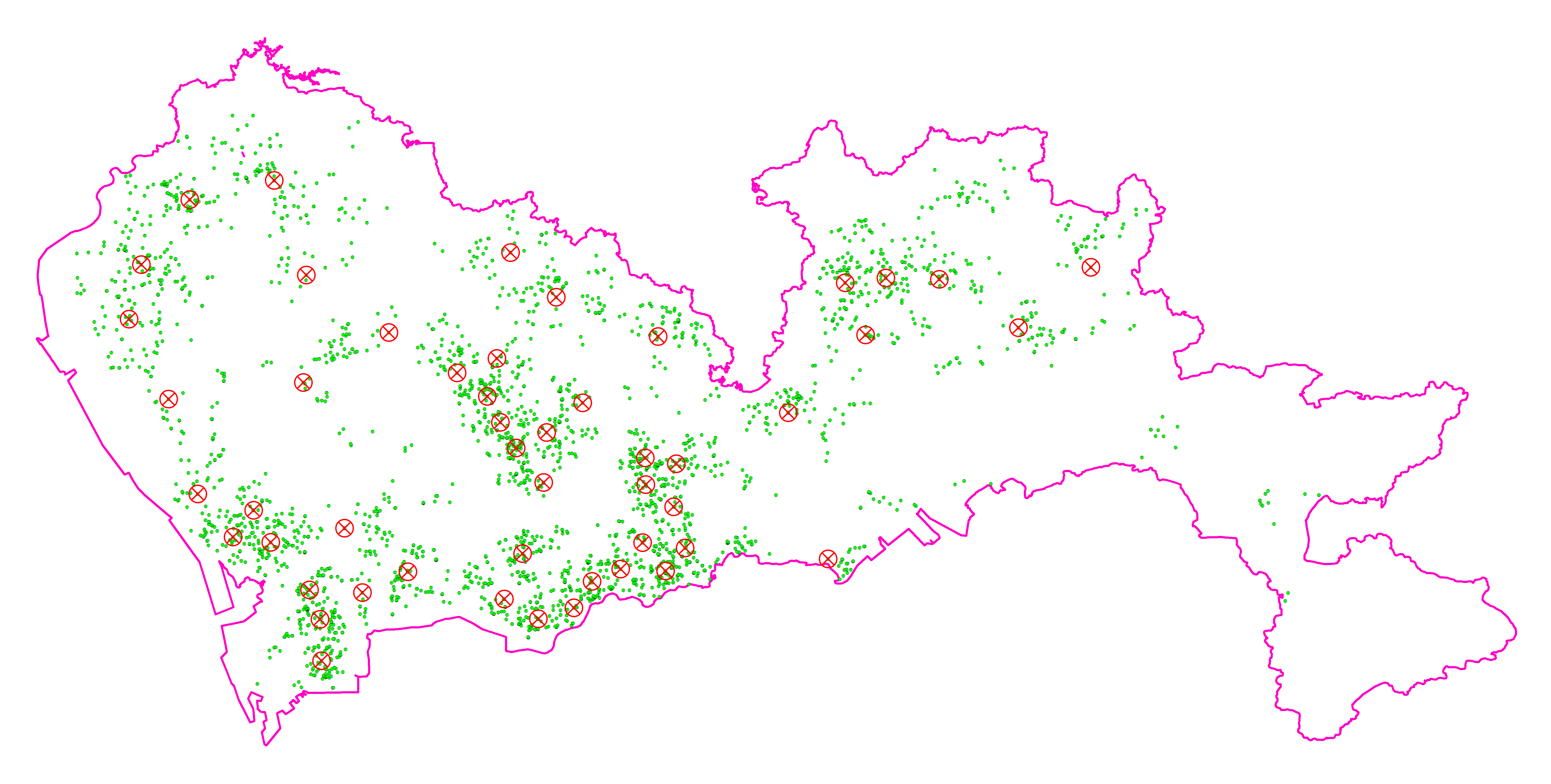}}
\caption{An illustration of an actual instance in Shenzhen, China. The small green circles are the $2936$ points of interest. The big red circles are the selected centers.}
\end{figure}

\section{Discussions}
In this paper, we propose and investigate the IF$k$CO, which overcomes the shortcoming of IF$k$C that the overall clustering quality may be effected by a few of isolated vertices.
As our main contribution, several approximation algorithms for the proposed problem are presented, with a provable $4$ approximation ratio. Despite the theoretical performance guarantee, the experiments on both synthetic and real-world datasets show that the refined algorithm usually outputs a feasible solution with performance significantly better than the approximation ratio of $4$.
To introduce the individual fairness and outlier detection into other clustering problems like the $k$-median and $k$-means are very interesting directions in the future.


\bibliographystyle{plain}
\bibliography{jogobibfile}

\begin{thebibliography}{10}

\bibitem{link}
{Gaode Maps}, {\url{https://lbs.amap.com/api/webservice/guide/api/ipconfig}}.

\bibitem{abv}
Mohsen Abbasi, Aditya Bhaskara, and Suresh Venkatasubramanian.
\newblock Fair clustering via equitable group representations.
\newblock In {\em Proceedings of the 2021 ACM Conference on Fairness,
  Accountability, and Transparency}, pages 504--514, 2021.

\bibitem{aekm}
Sara Ahmadian, Alessandro Epasto, Ravi Kumar, and Mohammad Mahdian.
\newblock Clustering without over-representation.
\newblock In {\em Proceedings of the 25th ACM SIGKDD International Conference
  on Knowledge Discovery \& Data Mining}, pages 267--275, 2019.

\bibitem{biosvw}
Arturs Backurs, Piotr Indyk, Krzysztof Onak, Baruch Schieber, Ali Vakilian, and
  Tal Wagner.
\newblock Scalable fair clustering.
\newblock In {\em Proceedings of the 36th International Conference on Machine
  Learning}, pages 405--413, 2019.

\bibitem{bcfn}
Suman~K Bera, Deeparnab Chakrabarty, Nicolas~J Flores, and Maryam Negahbani.
\newblock Fair algorithms for clustering.
\newblock In {\em Proceedings of the 33rd Annual Conference on Neural
  Information Processing Systems}, pages 4954--4965, 2019.

\bibitem{bgkkrss}
Ioana~O Bercea, Martin Gro{\ss}, Samir Khuller, Aounon Kumar, Clemens
  R{\"o}sner, Daniel~R Schmidt, and Melanie Schmidt.
\newblock On the cost of essentially fair clusterings.
\newblock {\em arXiv preprint arXiv:1811.10319}, 2018.

\bibitem{bprst}
Jaros{\l}aw Byrka, Thomas Pensyl, Bartosz Rybicki, Aravind Srinivasan, and Khoa
  Trinh.
\newblock An improved approximation for $k$-median, and positive correlation in
  budgeted optimization.
\newblock In {\em Proceedings of the 26th annual ACM-SIAM symposium on Discrete
  algorithms}, pages 737--756, 2014.

\bibitem{cgts}
Moses Charikar, Sudipto Guha, {\'E}va Tardos, and David~B Shmoys.
\newblock A constant-factor approximation algorithm for the $k$-median problem.
\newblock {\em Journal of Computer and System Sciences}, 65(1):129--149, 2002.

\bibitem{cklv}
Flavio Chierichetti, Ravi Kumar, Silvio Lattanzi, and Sergei Vassilvitskii.
\newblock Fair clustering through fairlets.
\newblock In {\em Proceedings of the 31st Annual Conference on Neural
  Information Processing Systems}, pages 5036--5044, 2017.

\bibitem{c}
Vincent Cohen-Addad.
\newblock Approximation schemes for capacitated clustering in doubling metrics.
\newblock In {\em Proceedings of the 14th Annual ACM-SIAM Symposium on Discrete
  Algorithms}, pages 2241--2259, 2020.

\bibitem{gsv}
Mehrdad Ghadiri, Samira Samadi, and Santosh Vempala.
\newblock Socially fair $k$-means clustering.
\newblock In {\em Proceedings of the 2021 ACM Conference on Fairness,
  Accountability, and Transparency}, pages 438--448, 2021.

\bibitem{g}
Teofilo~F Gonzalez.
\newblock Clustering to minimize the maximum intercluster distance.
\newblock {\em Theoretical computer science}, 38:293--306, 1985.

\bibitem{hs85}
Dorit~S Hochbaum and David~B Shmoys.
\newblock A best possible heuristic for the k-center problem.
\newblock {\em Mathematics of operations research}, 10(2):180--184, 1985.

\bibitem{hs86}
Dorit~S Hochbaum and David~B Shmoys.
\newblock A unified approach to approximation algorithms for bottleneck
  problems.
\newblock {\em Journal of the ACM}, 33(3):533--550, 1986.

\bibitem{hjv}
Lingxiao Huang, Shaofeng H-C Jiang, and Nisheeth~K Vishnoi.
\newblock Coresets for clustering with fairness constraints.
\newblock In {\em Proceedings of the 33rd Annual Conference on Neural
  Information Processing Systems}, pages 7589--7600, 2019.

\bibitem{jnn}
Matthew Jones, Huy Nguyen, and Thy Nguyen.
\newblock Fair $k$-centers via maximum matching.
\newblock In {\em International Conference on Machine Learning}, pages
  4940--4949, 2020.

\bibitem{jkl}
Christopher Jung, Sampath Kannan, and Neil Lutz.
\newblock Service in your neighborhood: Fairness in center location.
\newblock {\em Foundations of Responsible Computing}, 2020.

\bibitem{kps}
Samir Khuller, Robert Pless, and Yoram~J Sussmann.
\newblock Fault tolerant k-center problems.
\newblock {\em Theoretical Computer Science}, 242(1-2):237--245, 2000.

\bibitem{ks}
Samir Khuller and Yoram~J Sussmann.
\newblock The capacitated k-center problem.
\newblock {\em SIAM Journal on Discrete Mathematics}, 13(3):403--418, 2000.

\bibitem{kam}
Matth{\"a}us Kleindessner, Pranjal Awasthi, and Jamie Morgenstern.
\newblock Fair $k$-center clustering for data summarization.
\newblock In {\em Proceedings of the 36th International Conference on Machine
  Learning}, pages 3448--3457, 2019.

\bibitem{kls}
Ravishankar Krishnaswamy, Shi Li, and Sai Sandeep.
\newblock Constant approximation for $k$-median and k-means with outliers via
  iterative rounding.
\newblock In {\em Proceedings of the 50th Annual ACM SIGACT Symposium on Theory
  of Computing}, pages 646--659, 2018.

\bibitem{l}
Shi Li.
\newblock A 1.488 approximation algorithm for the uncapacitated facility
  location problem.
\newblock {\em Information and Computation}, 222:45--58, 2013.

\bibitem{mvi}
Sepideh Mahabadi and Ali Vakilian.
\newblock Individual fairness for $k$-clustering.
\newblock In {\em Proceedings of the 37th International Conference on Machine
  Learning}, pages 6586--6596. PMLR, 2020.

\bibitem{mv}
Yury Makarychev and Ali Vakilian.
\newblock Approximation algorithms for socially fair clustering.
\newblock {\em arXiv preprint arXiv:2103.02512}, 2021.

\bibitem{sss}
Melanie Schmidt, Chris Schwiegelshohn, and Christian Sohler.
\newblock Fair coresets and streaming algorithms for fair $k$-means.
\newblock In {\em Proceedings of the 17th International Workshop on
  Approximation and Online Algorithms}, pages 232--251, 2019.

\bibitem{sta}
David~B Shmoys, {\'E}va Tardos, and Karen Aardal.
\newblock Approximation algorithms for facility location problems.
\newblock In {\em Proceedings of the 29th annual ACM symposium on Theory of
  computing}, pages 265--274, 1997.

\bibitem{vy}
Ali Vakilian and Mustafa Yal{\c{c}}{\i}ner.
\newblock Improved approximation algorithms for individually fair clustering.
\newblock {\em arXiv preprint arXiv:2106.14043}, 2021.

\bibitem{xxdw}
Yicheng Xu, Dachuan Xu, Donglei Du, and Chenchen Wu.
\newblock Local search algorithm for universal facility location problem with
  linear penalties.
\newblock {\em Journal of Global Optimization}, 67(1-2):367--378, 2017.

\bibitem{xxzz}
Yicheng Xu, Dachuan Xu, Yong Zhang, and Juan Zou.
\newblock Mpuflp: Universal facility location problem in the $p$-th power of
  metric space.
\newblock {\em Theoretical Computer Science}, 838:58--67, 2020.

\end{thebibliography}

\end{document}